\documentclass[conference]{IEEEtran}
\IEEEoverridecommandlockouts
\usepackage{amsmath}
\usepackage{booktabs}
\usepackage[inline]{enumitem}
\usepackage[export]{adjustbox}
\usepackage{wrapfig}
\usepackage{subfigure}

\usepackage{amssymb,amsfonts} 
\usepackage{amsthm} 
\usepackage{commath} 
\usepackage{algorithmic}
\usepackage{mathtools, nccmath}

\usepackage{array}
\usepackage{textcomp}
\usepackage{xcolor}
\usepackage[final]{pdfpages}
\usepackage{comment}  

\PassOptionsToPackage{bookmarks=false}{hyperref}
\def\BibTeX{{\rm B\kern-.05em{\sc i\kern-.025em b}\kern-.08em
    T\kern-.1667em\lower.7ex\hbox{E}\kern-.125emX}}
    
\theoremstyle{definition}
\newtheorem{theorem}{Theorem}

\theoremstyle{definition}

\theoremstyle{definition}
\newtheorem{corollary}{Corollary}

\theoremstyle{definition}
\newtheorem{proposition}{Proposition}

\theoremstyle{definition}

\theoremstyle{definition}

\begin{document}

\title{Large-Signal Stability Guarantees for \\Cycle-by-Cycle Controlled DC-DC Converters}

\author{\IEEEauthorblockN{Xiaofan Cui and Al-Thaddeus Avestruz}
\thanks{*Xiaofan Cui and Al-Thaddeus Avestruz are with the Department
of Electrical Engineering and Computer Science, University
of Michigan, Ann Arbor, MI 48109, USA cuixf@umich.edu, avestruz@umich.edu.}
}
\IEEEoverridecommandlockouts
\IEEEpubid{\makebox[\columnwidth]{\hfill} \hspace{\columnsep}\makebox[\columnwidth]{ }}
\maketitle
\IEEEpubidadjcol

\begin{abstract}
\hspace{0pt} Stability guarantees are critical for cycle-by-cycle controlled dc-dc converters in consumer electronics and energy storage systems. Traditional stability analysis on \mbox{cycle-by-cycle} dc-dc converters is incomplete because the inductor current ramps are considered fixed; but instead, inductor ramps are not fixed because they are dependent on the output voltage in \mbox{large-signal} transients. 
We demonstrate a new \mbox{large-signal} stability theory which treats \mbox{cycle-by-cycle controlled} dc-dc converters as a particular type of feedback interconnection system.
An analytical and practical stability criterion is provided based on this system.
The criterion indicates that the $L/R$ and $RC$ time constants are the design parameters which determine the amount of coupling between the current ramp and the output voltage. 
\end{abstract}

\section{Introduction}
Cycle-by-cycle controlled dc\nobreakdash-dc converters are widely used in PoL\,(point-of-load) regulators \cite{Lee2013}, VRMs\,(voltage regulation modules) \cite{Svikovic2015}, battery chargers \cite{Yilmaz2013}, and LiDAR power supplies \cite{Cui2019c} because of its faster transient response compared to the traditional averaging-based control \cite{erickson2007}. However, the well\nobreakdash-known averaging theory cannot model the fast switching\nobreakdash-frequency\nobreakdash-scale dynamics of cycle\nobreakdash-by\nobreakdash-cycle controlled dc\nobreakdash-dc converters due to the slow-varying perturbation assumption. The switching-synchronized sampled\nobreakdash-state space (\emph{5S}) model is an accurate and tractable alternative \cite{Cui2018a}.

In comparison to the bilinear nonlinearity in the averaging model, the converter model in $\emph{5S}$ illustrates more complicated nonlinear behaviors.
While the small\nobreakdash-signal stability in $\emph{5S}$ has been widely discussed in \cite{Cui2018a}, the large\nobreakdash-signal stability of cycle-by-cycle controlled dc\nobreakdash-dc converters has not been adequately addressed. 

This paper focuses on one of the most widely-used cycle\nobreakdash-by\nobreakdash-cycle controlled dc\nobreakdash-dc converters --- current-mode dc-dc converters, which have a number of varieties including constant on\nobreakdash-time control \cite{Cui2019}, constant off\nobreakdash-time control \cite{xiaofanacc2019}, and fixed\nobreakdash-frequency peak current control \cite{Ding2020}.
In the existing large\nobreakdash-signal stability analysis for current\nobreakdash-mode converters, the current block and voltage block shown in Fig.\,\ref{fig:cmc_modeling} in current-mode dc-dc converters are considered decoupled and are designed separately.
The stability of the current block, which is referred to ``fast\nobreakdash-scale stability'', was studied by assuming the rising ramp and falling ramp of the inductor current are fixed as $m_1$ and $m_2$ \cite{Redl1981a}.
The stability of the voltage block, which is referred to ``slow\nobreakdash-scale stability'', was studied by utilizing averaging theory because of negligible voltage ripple and treating the current block as a controlled current source.

However, during large-signal transients, the output voltage significantly changes, hence the inductor current ramp changes every switching cycle. The cycle\nobreakdash-varying inductor current ramp affects the amount of charge pumped into the output, and ultimately affects the output voltage dynamics.
The traditional large\nobreakdash-signal stability analysis of the current block, which fully neglects the voltage block and assumes a fixed inductor current ramp, cannot guarantee the stability of the current block for large\nobreakdash-signal transients.

To address this deficiency in the large\nobreakdash-signal stability theory of current\nobreakdash-mode dc\nobreakdash-dc converters, we develop a new large\nobreakdash-signal stability theory which models the current\nobreakdash-mode buck converter as a feedback connection system in \emph{5S} shown in Fig.\,\ref{fig:cmc_modeling}.
Several discrete-time robust control tools, including small-gain theorem, dissipativity theory, and Lure system theory are utilized to rigorously study the stability of the resulting discrete-time nonlinear system.

This paper is organized as the following: 
(i) Section I introduces the paper; 
(ii) Section II develops the large-signal models for the current block and voltage block of a current-mode buck converter using constant on-time control;
(iii) the ultimate goal, which is illustrated in Section III, is to derive the large-signal stability guarantees for current mode dc-dc converters;
(iv) Section IV concludes the paper.
\begin{figure}[ht]
    \centering
    \includegraphics[width = 8cm]{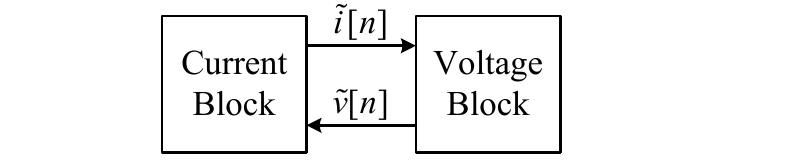}
    \caption{The current block and voltage block are coupled if the inductor current ramp is cycle-varying.}
    \label{fig:cmc_modeling}
\end{figure}
\section{\emph{5S} Modeling of Constant On-Time\\ Buck Converter}
\subsection{Converters, Systems, and Definitions}
Consider a class $\Sigma$ buck converter \cite{Cui2018a} using constant on-time current-mode control is illustrated in Fig.\,\ref{fig:schematicebuck_csl}.
\begin{figure}
    \centering
    \includegraphics[width = 7cm]{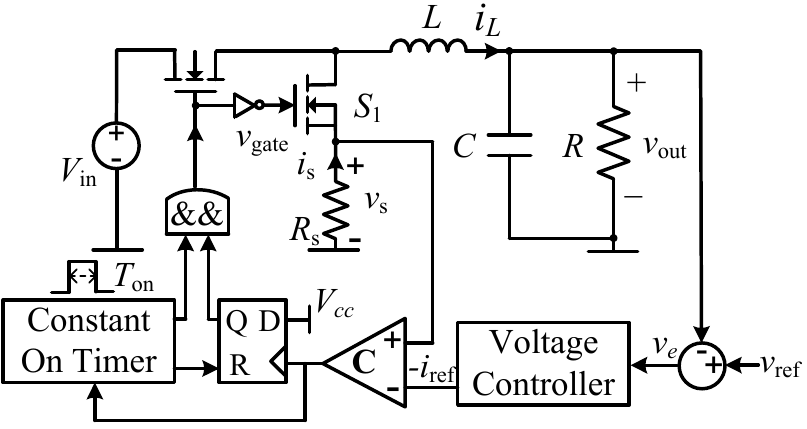}
    \caption{Schematic diagram of a digitally-controlled current-mode constant on-time buck converter.}
    \label{fig:schematicebuck_csl}
\end{figure}
The inductor current and capacitor voltage trajectories are shown in Fig.\,\ref{fig:cotcm_buck}. $V_{\text{in}}$ and $T_{\text{on}}$ are the input voltage and constant on time, respectively. According to the cycle-by-cycle control law in \cite{Cui2018a}, the output voltage is sampled during the on-time, once per switching cycle. The sampling time point for $v[n]$ can be expressed as the convex combination of the time of the inductor current valley $t_v[n]$ and the time for the inductor current peak $t_p[n]$.
\begin{figure}
    \centering
    \includegraphics[width = 6cm]{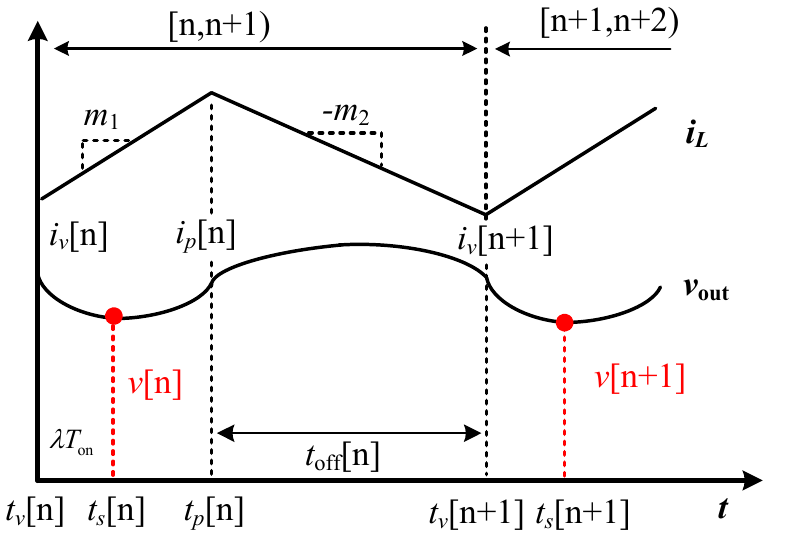}
    \caption{Constant on-time buck converter using cycle-by-cycle current-mode control.}
    \label{fig:cotcm_buck}
\end{figure}
\begin{align}
t_s[n] = \lambda t_v[n-1] + (1-\lambda)t_p[n],
\end{align}
the parameter $\lambda$ can be chosen to be from 0 to 1.

The slopes of the rising and falling ramps of the inductor current are denoted by $m_1[n]$ and $m_2[n]$
\begin{align} 
    m_1[n] &= \frac{V_{\text{in}}-v[n]}{L},  \quad m_2[n] = \frac{v[n]}{L}. \label{eqn:m1m2n}
\end{align}

We introduce the one-cycle-delayed valley current sequence $\{i_v^p[n]\}$ as 
\begin{align} 
i_v^p[n+1] \triangleq i_v[n]. 
\end{align}

We denote the equilibrium of the system by $v_{out}[n] = V_{\text{out}}$, $i_v[n] = I_{v}$, $i_v^p[n] = I_{v}$, and $t_\text{off}[n] = T_{\text{off}}$. The interference signal in the inductor current measurement is denoted by $w(t)$ \cite{cmpartone2022}. The equilibrium is defined by the following equations:
\begin{subequations}
\begin{align}
    I_v & \triangleq \frac{V_{\text{out}}}{R} - \frac{1}{2}\frac{V_{\text{in}}-V_{\text{out}}}{L}\,T_{\text{on}}, \quad I_v  = I_c - w(T_{\text{off}}), \label{eqn: equlibrium1}  \\
    T_{\text{off}} & \triangleq \frac{V_{\text{in}}-V_{\text{out}}}{V_{\text{out}}}\, T_{\text{on}}, \quad T_{s}^{\text{ss}} \triangleq T_{\text{on}} + T_{\text{off}} , \label{eqn: equlibrium2}
\end{align}
\end{subequations}
\vspace{-10pt}
\subsection{Current Block Modeling}
According to \cite{cmpartone2022}, the large-signal dynamical model (translated to the origin) of the current block follows
\begin{subequations} \label{eqn:buck_i_tilde_dyn}
\begin{align}
    \tilde{i}_v^{p}[n+1] & = \tilde{i}_v^{p}[n] - \frac{\tilde{v}[n]}{L}\,\,T_{\text{on}} - \frac{\tilde{v}[n]}{L}\,\,T_{\text{off}} - \frac{v[n]}{L}\,\,\tilde{t}_{\text{off}}[n], \\
    \tilde{i}_v^{p}[n+1]  & = \tilde{i}_c[n] - \psi(\tilde{t}_{\text{off}}[n]),
\end{align}
\end{subequations}
where the translated variables
$\tilde{v}[n]$, $\tilde{i}_v^p[n]$, $\tilde{t}_{\text{off}}[n]$, and $\psi$ satisfy
\begin{subequations}
    \begin{align}
        \tilde{v}[n] &= v[n] - V_{\text{out}}, \label{eqn:v_dev}  \quad
        \tilde{i}_v^p[n] = i_v^p[n] - I_v, \\
        \tilde{t}_{\text{off}}[n] &= t_{\text{off}}[n] - T_{\text{off}}, \quad
        \psi(x) = w(x+T_{\text{off}}) - w(x). \label{eqn:psi} 
    \end{align}
\end{subequations}
From (\ref{eqn:buck_i_tilde_dyn}), the current control block diagram can be represented by Fig.\,\ref{fig:Current_block}.
\begin{figure}[ht]
    \centering
    \includegraphics[width=7cm]{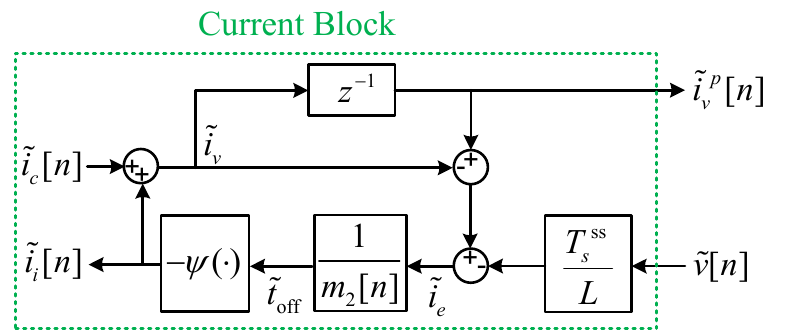}
    \caption{Current block diagram of current-mode buck converter using constant on-time.}
    \label{fig:Current_block}
\end{figure}
\vspace{-10pt}
\subsection{Voltage Block Modeling}
From \cite{Cui2018a}, constant on-time current-mode buck converters have the following two properties:
(1) the output $RC$-filter time constant is much greater than the switching period; 
(2) the output voltage has a small ripple so the inductor current can be considered as a cycle-varying piecewise linear (ramp) waveform. Property (1) implies
\begin{align} \label{eqn:assumption_1}
    \frac{T_s[n]}{RC} \ll 1,
\end{align}
where $T_s$ is the switching period at the $n^{\text{th}}$ switching cycle. Property (2) implies that the {\em{quasi-steady state discharging}} of the capacitor, i.e. the discharging current to the load, can be treated as constant throughout a given switching cycle, and the {\em{small output-voltage ripple}}
\begin{align} \label{eqn:assumption_2}
    \frac{T_{s}[n]\,T_{\text{on}}}{2LC} & \ll 1.
\end{align}
The large-signal dynamical model (translated to the origin) of the voltage block follows \cite{Cui2018a}:
\begin{subequations} \label{eqn:buck_tilde_dyn}
\begin{align}
    \tilde{v}[n+1] = \,\,& \tilde{v}[n] + \frac{1}{C}\left(\sum_{j=1}^3\tilde{Q}^{(j)}_{\text{in}}[n] - \tilde{Q}_{\text{out}}[n]\right), \label{eqn:tildev_dyn} \\
    \tilde{Q}^{(1)}_{\text{in}}[n] = \,\,& (1-\lambda) T_{\text{on}} \tilde{i}_v^p[n] - \frac{1 - \lambda^2}{2} \frac{\tilde{v}[n]}{L}\; T^2_{\text{on}}, \label{eqn:tildeqin1}\\
    \tilde{Q}^{(2)}_{\text{in}}[n] = \,\,& \frac{1}{2}\left( \tilde{i}_v^p[n] -\frac{\tilde{v}[n]}{L}\;T_{\text{on}} +  \tilde{i}_v^p[n+1] \right)t_{\text{off}}[n] \nonumber \\
    &+ \left(I_v + \frac{M_1T_{\text{on}}}{2}\right)\tilde{t}_{\text{off}}[n], \label{eqn:tildeqin2} 
\end{align}
\begin{align}
    \tilde{Q}^{(3)}_{\text{in}}[n] = \,\,& \lambda T_{\text{on}}  i_v^p[n+1] - \frac{\lambda^2}{2}  \frac{\tilde{v}[n]}{L} \; T^2_{\text{on}},  \label{eqn:tildeqin3}\\
    \tilde{Q}_{\text{out}}[n] = \,\,& \frac{\tilde{v}[n]}{R}\left( T_{\text{on}} + t_{\text{off}}[n] \right) +\frac{V_{\text{out}}}{R} \;\tilde{t}_{\text{off}}[n]. \label{eqn:tildeqout}
\end{align}
\end{subequations}
To prevent the switching transient from disturbing the valley current detection denoted by the sense voltage $v_s$ in Fig.\,\ref{fig:schematicebuck_csl}, the time-varying off-time is bounded from below \cite{Cui2018a}. To avoid the misdetection of the valley current event, the time-varying off-time is bounded from above by
\begin{align} \label{eqn:toff_bd}
    T^{\text{min}}_{\text{off}} \le t_{\text{off}}[n] \le T^{\text{max}}_{\text{off}}.
\end{align}
\section{Stability and Control Performance Analysis of Constant On-Time Buck Converters}
\subsection{Current Block}
To calculate the $\mathcal{L}_2$ gain from $\tilde{v}[n]$ to $\tilde{i}_v^p[n]$, denoted by $\Gamma_{v \to i}$, we introduce the following dissipativity theory.
\begin{figure}
    \centering
    \includegraphics[width = 3.5cm]{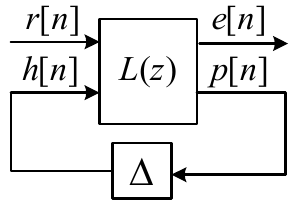}
    \caption{Linear fractional transformed system.}
    \label{fig:Lft_general}
\end{figure}
If the system can be expressed in $F_u(L,\Delta)$ form \cite{Boyd1994}, as shown in Fig.\,\ref{fig:Lft_general}, then
the state-space representation of $F_u(L,\Delta)$ is
\begin{subequations}
    \begin{align}
        x[n+1] & = A x[n] + B_1 h[n] + B_2 r[n], \\
        p[n] & = C_1 x[n] + D_{11} h[n] + D_{12} r[n], \\
        e[n] & =  C_2 x[n] + D_{21} h[n] + D_{22} r[n], \\
        h[n] & = \Delta \, \left( p[n] \right),
    \end{align}
\end{subequations}
the upper bound of the $\mathcal{L}_2$ gain can be calculated as
\begin{theorem} \label{theorem:lure_gain}
Assume $\Delta: \mathbb{R} \rightarrow \mathbb{R}$ is an $[\hat{\alpha}, \hat{\beta}]$ sector-bounded nonlinearity. Also, assume $F_u(L,\Delta)$ is well-posed. If $\exists \, P > 0$, $\lambda \ge 0$, and $\hat{\gamma} > 0$ such that
\begin{align} \label{eqn:LMI_constraint}
    & \begin{bmatrix}
        A^TPA - P &   A^TPB_1 &   A^TPB_2 \\
        B_1^TPA                & B_1^TPB_1 &   B_1^TPB_2 \\
        B_2^TPA                & B_2^TPB_1 &   B_2^TPB_2 - \hat{\gamma}^2 I
    \end{bmatrix} \nonumber \\
    & + \lambda \begin{bmatrix}
    C_1 & D_{11} & D_{12}\\
    0   & I      & 0  
    \end{bmatrix}^T 
    \begin{bmatrix}
    -\hat{\alpha} \hat{\beta} &  \frac{\hat{\alpha} + \hat{\beta}}{2} \\ 
    \frac{\hat{\alpha} + \hat{\beta}}{2} & -1
    \end{bmatrix} 
    \begin{bmatrix}
    C_1 & D_{11} & D_{12}  \\
    0   & I      & 0
    \end{bmatrix} \nonumber \\
    & + \begin{bmatrix}
    C_2^T \\
    D_{21}^T \\
    D_{22}^T 
    \end{bmatrix} 
    \begin{bmatrix}
    C_2 & D_{12} & D_{22}
    \end{bmatrix}
    <0,
\end{align}
then 
$\|F_u(L,\Delta)\| \le \hat{\gamma}$.
\end{theorem}
\begin{proof}
Let $r \in L_2$ be any input and assume $x[0] = 0$. Let $(x, p, h, e)$ be the resulting solutions for this input $r$.

Multiply the left and right of (\ref{eqn:LMI_constraint}) by $\left[x[n]^T\,\,p[n]^T\,\,r[n]^T\right]$ and its transpose to obtain
\begin{align} \label{eqn:dissipation_ineqn}
& V(x[n+1]) - V(x[n]) - \hat{\gamma}^2 r[n]^Tr[n] \nonumber \\ 
&+ \Lambda \begin{bmatrix}
    p[n] \\
    h[n] 
    \end{bmatrix}^T
    \begin{bmatrix}
    -\hat{\alpha} \hat{\beta} &  \frac{\hat{\alpha} + \hat{\beta}}{2} \\ 
    \frac{\hat{\alpha} + \hat{\beta}}{2} & -1
    \end{bmatrix} 
    \begin{bmatrix}
        p[n] \\
        h[n] 
        \end{bmatrix} + e[n]^Te[n] < 0,
\end{align}
where the storage function $V(x)$ is defined as $V(x)\triangleq x^T P x$.

By summing up both sides of (\ref{eqn:dissipation_ineqn}) from $n=0$ to $n = N$, we have
\begin{align} \label{eqn:dissipation_ineqn_sum}
& V(x[N+1]) - V(x[0]) - \hat{\gamma}^2 \sum_{n=0}^{N} r[n]^Tr[n] \nonumber \\ 
&+ \Lambda \sum_{n=0}^{N} \begin{bmatrix}
    p[n] \\
    h[n] 
    \end{bmatrix}^T
    \begin{bmatrix}
    -\hat{\alpha} \hat{\beta} &  \frac{\hat{\alpha} + \hat{\beta}}{2} \\ 
    \frac{\hat{\alpha} + \hat{\beta}}{2} & -1
    \end{bmatrix} 
    \begin{bmatrix}
        p[n] \\
        h[n] 
        \end{bmatrix} + \sum_{n=0}^{N} e[n]^Te[n] < 0.
\end{align}
(\ref{eqn:dissipation_ineqn_sum}) can be equivalently expressed as 
\begin{align} \label{eqn:dissipation_ineqn_sum2}
& V(x[N+1]) - V(x[0]) + \sum_{n=0}^{N} e[n]^Te[n] -\sum_{n=0}^{N} \hat{\gamma}^2 r[n]^Tr[n] \nonumber \\ 
&+ \Lambda \sum_{n=0}^{N}\left(h[n] - \hat{\alpha} p[n]\right)\left(\hat{\beta} p[n] - h[n]\right) < 0.
\end{align}
$V(x[N+1]) \ge 0$ because $P>0$. $V(x[0])=0$ because $x[0]=0$. $h[n] \ge \hat{\alpha} p[n]$ and $\hat{\beta} p[n] \le h[n]$ because of the sector-bounded nonlinearity. 
The last term in (\ref{eqn:dissipation_ineqn_sum2}) is non-negative because of both the sector-bounded nonlinearity and positive scalar term $\Lambda$. Therefore, (\ref{eqn:dissipation_ineqn_sum2}) implies
\begin{align}
    \sum_{n=0}^{N} e[n]^Te[n] <\sum_{n=0}^{N} \hat{\gamma}^2 r[n]^Tr[n].
\end{align}

Take $N \rightarrow \infty$ to show $\|e\|_2^2 < \hat{\gamma}^2 \|r\|_2^2$. In all, the $\mathcal{L}_2$ norm of $F_u(L,\Delta)$ is bounded by $\hat{\gamma}$.
\end{proof}
We reformulate the current block in a unitless $F_u(L, \Delta)$ form as
\begin{subequations} \label{eqn:unitless_current-block}
\begin{align}
    \tilde{i}_v^{p}[n+1] &= \tilde{i}_i[n], \quad
    \tilde{i}_e[n] = \tilde{i}_v^{p}[n] - \tilde{i}_i[n] + \tilde{u}[n], \\
    \tilde{i}_v^{p}[n] &= \tilde{i}_v^{p}[n], \quad 
     \tilde{i}_i[n] = \Delta (\tilde{i}_e[n]),
\end{align}
\end{subequations}
where
\begin{align}
     \tilde{u}[n] = \frac{T^{\text{ss}}_{s}}{L} \tilde{v}[n],\quad
     \Delta(z) = -\psi \left(\frac{z}{m_2[n]} \right), 
\end{align}
and $\Delta(z) \in [\hat{\alpha}, \hat{\beta}]$ is a sector-bounded time-varying nonlinearity. 

From Theorem \ref{theorem:lure_gain}, the gain $\hat{\gamma}$ of unitless $F_u(L, \Delta)$ form can be obtained from following optimization problem:
\begin{subequations} \label{optproblem:currentgain}
\begin{align}
\underset{\hat{\gamma}, \lambda, P}{\mathrm{min}}& \,\,\,  \hat{\gamma}^2  \\
\text{subject to} & \,\, P>0, \,  \lambda\ge 0,\,  \hat{\gamma} >0 \label{eqn:opt_layer1_contraint1},\\
    &  \lambda \begin{bmatrix}
        - \hat{\alpha}\hat{\beta} & \hat{\alpha}\hat{\beta} + \frac{\hat{\alpha}+\hat{\beta}}{2}  &  -\hat{\alpha}\hat{\beta} \\
        \hat{\alpha}\hat{\beta} + \frac{\hat{\alpha}+\hat{\beta}}{2}   & - 2\hat{\alpha}\hat{\beta} -\hat{\alpha}-\hat{\beta} & \hat{\alpha}\hat{\beta} + \frac{\hat{\alpha}+\hat{\beta}}{2}    \\
        -\hat{\alpha}\hat{\beta}   & \hat{\alpha}\hat{\beta} + \frac{\hat{\alpha}+\hat{\beta}}{2} &  -\hat{\alpha}\hat{\beta}
    \end{bmatrix} \nonumber \\
    & + \begin{bmatrix}
        1 - P &  0 &  0 \\
        0  & P &   0 \\
        0   & 0 &  - \hat{\gamma}^2
    \end{bmatrix} < 0. \label{eqn:opt_layer1_contraint2}
\end{align}
\end{subequations}
Problem (\ref{optproblem:currentgain}) is in linear matrix inequality (LMI) form, hence is convex and the global minimum exists. By applying an LMI solver (e.g. CVX), we obtain the gain of system (\ref{eqn:unitless_current-block}) as a function of the sector bounds,
\begin{align}  \label{eqn:def_gab}
    \hat{\gamma} = g(\hat{\alpha}, \hat{\beta}).
\end{align}
If there is no interference, \mbox{$\hat{\alpha} = \hat{\beta} = 0$}, system (\ref{eqn:unitless_current-block}) has zero gain, and voltage doesn't affect the current. As $\hat{\alpha}$ decreases and $\hat{\beta}$ increases, system (\ref{eqn:unitless_current-block}) has larger gain, as shown in Fig.\,\ref{fig:Current_block_gain}.

\begin{figure}
    \centering
    \includegraphics[scale = 1]{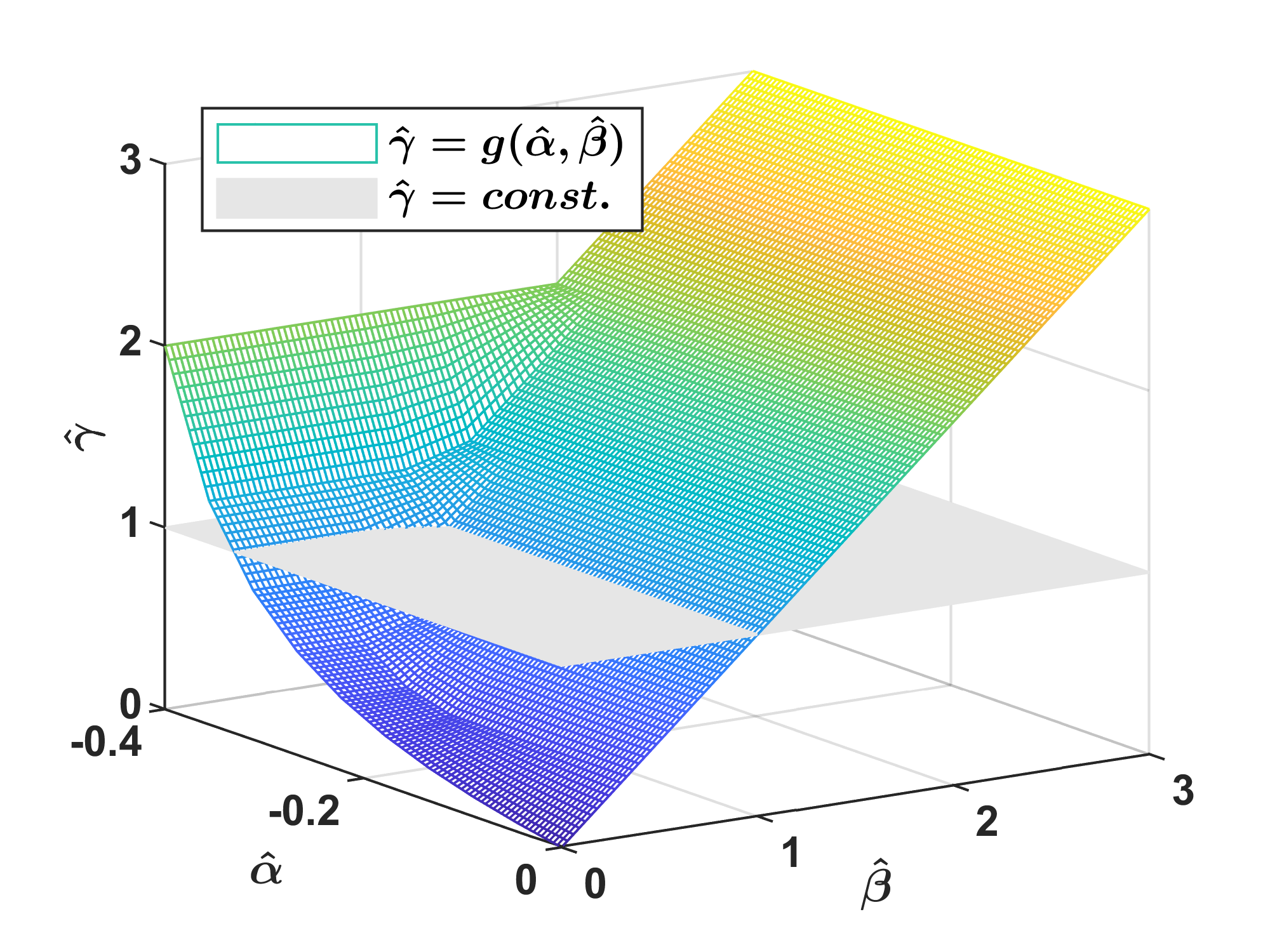}
    \caption{Gain of system (\ref{eqn:unitless_current-block}) as a function of sector bounds. For no interference, $\hat{\alpha} = \hat{\beta} = 0$, system (\ref{eqn:unitless_current-block}) has zero gain, and voltage does not affect the current. As $\hat{\alpha}$ decreases and $\hat{\beta}$ decreases, system (\ref{eqn:unitless_current-block}) has larger gain.}
    \label{fig:Current_block_gain}
\end{figure}
The following corollary follows Theorem\, \ref{theorem:lure_gain}:\\
\begin{corollary}
Given the class $\Sigma$ buck converter modeled by \cite{Cui2018a}, the $\mathcal{L}_2$ gain from the sampled output voltage sequence $\{\tilde{v}[n]\}$ to one-cycle-delayed inductor current sequence $\{\tilde{i}^{p}_v[n]\}$ is bounded from above by
\begin{align}\label{eqn:gamma_vi}
    \Gamma_{v \to i} \le \frac{T_s^{\text{ss}}}{L} g(\hat{\alpha}, \hat{\beta}),
\end{align}
where $T_s^{\text{ss}}$ is the steady-state switching period, 
\begin{align}
    T_s^{\text{ss}} & \triangleq T_{\text{on}} +  T_{\text{off}}.
\end{align}
\end{corollary}
\subsection{Voltage Block}
The voltage dynamics, which is described by the nonlinear time-invariant system (\ref{eqn:buck_tilde_dyn}), is a nonlinear system with quadratic and fractional nonlinearities. A straightforward mathematical tool does not exist to analytically study the stability and calculate the $\mathcal{L}_2$ gain in closed form. The algorithmic method only applies to a buck converter with specific $L$, $C$, and $R$. Therefore, it is difficult to generalize the algorithmic methods and give engineering intuitions to the designers.

We solve this challenge by taking advantage of the inherent physical constraint of constant on-time buck converters that the cycle-varying off-time is bounded by (\ref{eqn:toff_bd}).
\begin{theorem}
Given the class $\Sigma$ buck converter modeled by \cite{Cui2018a}, the $\mathcal{L}_2$ gain from the one-cycle-delayed inductor current sequence $\{\tilde{i}^{p}_v[n]\}$ to sampled output voltage sequence $\{\tilde{v}[n]\}$ is bounded from above by
\begin{align}
    \Gamma_{i \to v} &  \le \frac{R}{\left( 1 + \frac{T_{\text{on}}}{2\tau_2}\right)} \frac{T_s^{\text{max}}}{T_s^{\text{min}}},
\end{align}
where $T_s^{\text{min}}$ and $T_s^{\text{max}}$ are the shortest switching period and longest switching period, respectively, and $\tau_2$ is the $L/R$ time constant,
\begin{align}
    T_s^{\text{max}} & \triangleq T_{\text{on}} +  T_{\text{off}}^{\text{max}}, \quad
    T_s^{\text{min}}  \triangleq T_{\text{on}} +  T_{\text{off}}^{\text{min}}, \quad
    \tau_2 \triangleq \frac{L}{R}.
\end{align}
\end{theorem}

\begin{proof}
(i) \emph{Voltage Block Model Reformulation}\\
The nonlinear time-invariant system (\ref{eqn:buck_tilde_dyn}) can be transformed to the following linear time-varying system
\begin{subequations}
\begin{align}
       \tilde{v}[n+1] &= \alpha[n]\,\tilde{v}[n] +  \beta[n]\,\tilde{i}_v^p[n] + \gamma[n]\,\tilde{i}_v^p[n+1] ,\\
       \alpha[n] &= 1-\frac{T_{\text{on}}+t_{\text{off}}[n]}{RC}-\frac{T_{\text{on}}\left( T_{\text{on}} + t_{\text{off}}[n] \right)}{2LC},\\
       \beta[n] &= \frac{1}{C} \left( (1-\lambda ) T_{\text{on}} + \frac{1}{2} t_{\text{off}}[n] \right), \\
       \gamma[n] &= \frac{1}{C} \left( \lambda T_{\text{on}} + \frac{1}{2} t_{\text{off}}[n] \right).
    \end{align}
\end{subequations}
From (\ref{eqn:toff_bd}), the time-varying coefficients are bounded by
\begin{subequations}
\begin{align}
     0 < \alpha[n] & \le \alpha_{\text{max}} = 1 - \frac{T_{s}^{\text{min}}}{RC} - \frac{T_{\text{on}} T_{s}^{\text{min}}}{2LC},\\
     0 < \beta[n] & \le \beta_{\text{max}} = \frac{1}{C} \left( (1-\lambda) T_{\text{on}} + \frac{1}{2} T_{\text{off}}^{\text{max}} \right), \\
     0 < \gamma[n] & \le \gamma_{\text{max}} = \frac{1}{C} \left( \lambda T_{\text{on}} + \frac{1}{2} T_{\text{off}}^{\text{max}}  \right).
\end{align}
\end{subequations}

We transform the linear time-varying system into standard state-space representation
\begin{subequations}
\begin{align}
    \tilde{v}[n+1] &= \alpha[n]q[n] + \left( \beta[n]+ \alpha[n] \gamma[n] \right) \tilde{i}_v^p[n], \label{eqn:ivq2q} \\ 
    \tilde{v}[n] &=  q[n] + \gamma[n] \tilde{i}_v^p[n].  \label{eqn:ivq2v}
\end{align}
\end{subequations}
(ii) \emph{Gain Estimation} \\
We utilize a storage function $\mathcal{V}[n] = q^2[n]$ to calculate the $\mathcal{L}_2$ gain of the dynamical system (\ref{eqn:ivq2q})
\begin{align} \label{eqn:storage_func_diff}
   \mathcal{V}[n+1]- \mathcal{V}[n] = q^2[n+1] - q^2[n].
\end{align}
From the Cauchy-Schwartz Inequality, given any $\mu_1 > 0$, $ q^2[n+1]$ can be bounded from the above by
\begin{align}
 & q^2[n+1] \nonumber\\
=& \,\, \left( \alpha[n] q[n] + \left(\beta[n] + \alpha[n] \gamma[n]\right) \tilde{i}_v^p[n] \right)^2 \nonumber\\
\le& \,\,  \left( \alpha^2 [n] q^2[n] + (\beta[n] + \alpha[n] \gamma[n])^2
\frac{(\tilde{i}_v^p[n])^2}{\mu_1}\right) \left(1+\mu_1\right).
\end{align}
We let \mbox{$\mu_1 = (1 - \alpha[n]) (\alpha[n])^{-1}$}. It can be verified that \mbox{$\mu_1 > 0$} from (\ref{eqn:assumption_1}) and (\ref{eqn:assumption_2})
\begin{align} \label{eqn:qn_bd}
  & q^2[n+1] \le \alpha[n] q^2[n] + \frac{\left(\beta[n] + \alpha[n] \gamma[n]\right)}{1- \alpha[n]} (\tilde{i}_v^p[n])^2.
\end{align}
Substitute (\ref{eqn:qn_bd}) to (\ref{eqn:storage_func_diff}),
\begin{align}
    & \mathcal{V}[n+1]- \mathcal{V}[n] \le \nonumber \\
& \frac{\left(\beta[n] + \alpha[n] \gamma[n]\right)^2}{1- \alpha[n]} \; (\tilde{i}_v^p[n])^2 - (1 - \alpha[n]) q^2[n] \le  \nonumber\\
 & \frac{\left(\beta_{\text{max}} + \alpha_{\text{max}} \gamma_{\text{max}}\right)^2}{1- \alpha_{\text{max}}}\; (\tilde{i}_v^p[n])^2 - (1 - \alpha_{\text{max}}) q^2[n] =  \nonumber \\
& (1 - \alpha_{\text{max}}) \left( \frac{\left(\beta_{\text{max}} + \alpha_{\text{max}} \gamma_{\text{max}}\right)^2}{(1- \alpha_{\text{max}})^2} \; (\tilde{i}_v^p[n])^2 -  q^2[n] \right).
\end{align}
Summing both sides of the inequality for all $n$ yields
\begin{align}
    \mathcal{V}(\infty)- \mathcal{V}(1) \le (1 - \alpha_{\text{max}})\left(\Gamma_1^2 \|\tilde{i}_v^p\|_2^2 - \|q\|_2^2 \right),
\end{align}
where $\|\cdot\|_2$ is the $\mathcal{L}_2$ norm and 
\begin{align}
    \Gamma_1 = \frac{\left(\beta_{\text{max}} + \alpha_{\text{max}} \gamma_{\text{max}}\right)}{(1- \alpha_{\text{max}})}. \label{eqn:gamma_1_def}
\end{align}
The $\mathcal{L}_2$ norm of $q[n]$ can be bounded by
\begin{align}
    \|q\|^2 & \le \Gamma_1^2  \|\tilde{i}_v^p\|_2^2 + \frac{V(1)-V(\infty)}{(1-\alpha_{\text{max}})} \le \Gamma_1^2 \|\tilde{i}_v^p\|_2^2 + \frac{V(1)}{(1-\alpha_{\text{max}})}. \label{eqn:gain_q_gamma_1}
\end{align}
By definition, the $\mathcal{L}_2$ gain of system ($\ref{eqn:ivq2q}$) is bounded from above by $\Gamma_1$.

The $\mathcal{L}_2$ gain of the system (\ref{eqn:ivq2v}) can be  obtained from the Cauchy-Schwartz Inequality as 
\begin{align}
     & \tilde{v}^2[n] = (q[n] + \gamma[n] \tilde{i}_v^p[n])^2 \le \nonumber \\
     & \left( q^2[n] + \gamma^2[n] (\tilde{i}_v^p[n])^2 \frac{\Gamma_1}{\gamma[n]}\right)\left(1+\frac{\gamma[n]}{\Gamma_1}\right) \le \nonumber \\
      & \left(\Gamma_1 + \gamma[n]\right) \Gamma_1(\tilde{i}_v^p[n])^2 +   \left(\Gamma_1 + \gamma[n]\right) \gamma[n](\tilde{i}_v^p[n])^2 \le  \nonumber\\
      & \left(  \Gamma_1 + \gamma_{\text{max}} \right) (\tilde{i}_v^p[n])^2.
\end{align}
Summing both sides of the inequality for all $n$ yields the $\mathcal{L}_2$ gain of the system (\ref{eqn:ivq2v})
\begin{align}
    \|v\|_2 \le \left(  \Gamma_1 + \gamma_{\text{max}} \right) \|\tilde{i}_v^p\|_2. \label{eqn:gain_i_v}
\end{align}
By (\ref{eqn:gamma_1_def}), the $\mathcal{L}_2$ gain of the voltage block can be bounded from above 
\begin{align} \label{egn:gamma_iv}
    \Gamma_{i \to v} & \le \Gamma_1 + \gamma_{\text{max}} = \frac{\beta_{\text{max}}+\gamma_{\text{max}}}{1-\alpha_{\text{max}}} =  \frac{R}{\left( 1 + \frac{T_{\text{on}}}{2\tau_2}\right)} \frac{T_s^{\text{max}}}{T_s^{\text{min}}},
\end{align}
where $T_s^{\text{min}}$ and $T_s^{\text{max}}$ are the shortest switching period and longest switching period, respectively, and $\tau_2$ is the $L/R$ time constant.
\end{proof}

\subsection{Overall System}
From the Small Gain Theorem \cite{Okuyama2014}, the current-mode buck converter is finite-gain $\mathcal{L}_2$ stable if
\begin{align} \label{eqn:sgt_sys}
    \Gamma_{i\to v} \cdot \Gamma_{v\to i} < 1.
\end{align}
By substituting (\ref{eqn:gamma_vi}) and (\ref{egn:gamma_iv}) into (\ref{eqn:sgt_sys}), (\ref{eqn:sgt_sys}) is equivalent to
\begin{align} \label{eqn:stability_criterion}
    g(\hat{\alpha}, \hat{\beta}) < \left(\tau_2 + \frac{T_{\text{on}}}{2}\right)\frac{T_s^{\text{min}}}{T_s^{\text{max}}} \frac{1}{T_s^{ss}},
\end{align}
where $T_s^{\text{min}}$, $T_s^{\text{max}}$, and $T_s^{\text{ss}}$ are the shortest switching period, longest switching period, and steady-state switching period, respectively.

Because both current and voltage blocks are zero state observable \cite{khalil2002nonlinear}, finite-gain $\mathcal{L}_2$ stability implies the current-mode buck converter is large-signal asymptotically stable. The following corollary follows from Corollary 2:
\begin{corollary}
\label{theore: buck_ontime_stability}
The current control loop of class $\Sigma$ buck converters using constant on-time current-mode control is globally asymptotically stable if 
\begin{align}
    & g(\hat{\alpha}, \hat{\beta}) < \left(\tau_2 + \frac{T_{\text{on}}}{2}\right)\frac{T_s^{\text{min}}}{T_s^{\text{max}}} \frac{1}{T_s^{ss}},\\
   & T_s^{\text{max}}\left( 1 + \frac{T_{\text{on}}}{2\tau_2}\right) < \tau_1,
\end{align}
where $g(\hat{\alpha}, \hat{\beta})$ follows (\ref{eqn:def_gab}), $T_s^{\text{min}}$ and $T_s^{\text{max}}$ are the shortest switching period and longest switching period, respectively,
$\tau_1$ is the $RC$ time constant, and $\tau_2$ is the $L/R$ time constant,
\begin{align}
    T_s^{\text{max}} & \triangleq T_{\text{on}} +  T_{\text{off}}^{\text{max}}, \quad
    T_s^{\text{min}} \triangleq T_{\text{on}} +  T_{\text{off}}^{\text{min}},\\
    \tau_1 &\triangleq RC, \quad \tau_2 \triangleq \frac{L}{R}.
\end{align}
\end{corollary}

\section{Modeling and Stability of Constant Off-Time Boost Converter}
The modeling and stability of the constant off-time boost converters can be studied in the similar manners. 

\subsection{Current-Mode Boost Converter Using Constant Off-Time}
Consider a class $\Sigma$ boost converter using constant off-time defined in \cite{Cui2018a}, the on-time in steady state is denoted by $T_{\text{off}}$.
The time-varying on-time is bounded from above by
\begin{align} \label{eqn:toff_bd_boost}
    T^{\text{min}}_{\text{on}} \le t_{\text{on}}[n] \le T^{\text{max}}_{\text{on}}.
\end{align}
The large-signal stability guarantee is shown in Proposition\,\ref{theore: boost_offtime_stability_cotcmboost}. The detailed proof can be found in \cite{Avestruz2022}:

\begin{proposition}\label{theore: boost_offtime_stability_cotcmboost}
The current control loop of the class $\Sigma$ boost converter using constant off-time control is globally asymptotically stable

\vspace{+8pt}
\noindent (i) if $\left((1-\lambda ) T_{\text{off}} + \frac{V_{\text{out}}L}{V_{\text{in}}R} \right)\left( 1-\frac{T_s^{\text{ss}}}{RC} -\frac{T_{s}^{\text{max}}}{RC}-\frac{ T^2_{\text{off}}}{2LC}\right) + \left( \lambda T_{\text{off}} - \frac{V_{\text{out}}L}{V_{\text{in}}R} \right) \ge 0 $\:\:and
\begin{align}
    g(\hat{\alpha}, \hat{\beta}) &  \le 
    \frac{1}{2} + \tau_2
   \left(\frac{T_s^{\text{ss}}+T_s^{\text{min}}}{T_s^{\text{ss}}T_{\text{off}}} \right);
\end{align}

\noindent or (ii) if $\left((1-\lambda ) T_{\text{off}} + \frac{V_{\text{out}}L}{V_{\text{in}}R} \right)\left( 1-\frac{T_s^{\text{ss}}}{RC} -\frac{T_{s}^{\text{max}}}{RC}-\frac{ T^2_{\text{off}}}{2LC}\right) + \left( \lambda T_{\text{off}} - \frac{V_{\text{out}}L}{V_{\text{in}}R} \right) < 0 $\:\:and
\begin{align}
    &g(\hat{\alpha}, \hat{\beta}) \le\\\nonumber
    &\frac{2\tau_2\left( T_s^{\text{min}} +  T_s^{\text{ss}}\right) + T_{\text{off}}^2}{2\tau_2\left( T_s^{\text{max}} +  T_s^{\text{ss}}\right) + T_{\text{off}}^2} \frac{2\tau_1-T_s^{\text{ss}}-T_s^{\text{max}}-\frac{T_{\text{off}}^2}{2\tau_2}}{2\frac{V_{\text{out}}}{V_{\text{in}}}+(1-2\lambda)\frac{T_{\text{off}}}{\tau_2}}T_{\text{off}};
\end{align}
where $g(\hat{\alpha}, \hat{\beta})$ follows Fig.\,\ref{fig:Current_block_gain}, $T_s^{\text{min}}$ and $T_s^{\text{max}}$ are the shortest switching period and longest switching period, respectively, $\tau_1$ is the $RC$ time constant, and $\tau_2$ is the $L/R$ time constant,
\begin{align}
    T_s^{\text{max}} & \triangleq T_{\text{on}} +  T_{\text{off}}^{\text{max}}, \\
    T_s^{\text{min}} & \triangleq T_{\text{on}} +  T_{\text{off}}^{\text{min}},\\
    \tau_1 &\triangleq RC, \quad 
    \tau_2 \triangleq \frac{L}{R}.
\end{align}
\end{proposition}

\section{Simulation Validation}
The proposed stability criterion was validated by a high-fidelity switch-circuit simulation model in MATLAB/Simulink.
The parameters of the constant on-time buck converter system in the validation are shown in Table\,\ref{table:cotcm_buck_param}.
Two case studies were performed with different resistive load conditions as shown in Table\,\ref{table:cont_on_cm_stability_criterion}.
The inductor currents were controlled so that the steady-state output voltages were kept the same.
The inadequacy of the existing stability criterion \cite{cmpartone2022} is illustrated for Case 1 in Fig.~ \ref{fig:alpha_0_48.png}.
The proposed stability criterion,  which guarantees the large-signal stability of the power converter system, is demonstrated for Case 2 in Fig.~\ref{fig:alpha_0_30.png}. 


\begin{table}[tb]
    \caption{Design Parameters of the Constant On\nobreakdash-Time Current\nobreakdash-Mode Buck Converter}
    \label{table:cotcm_buck_param}
    \centering
    \begin{tabular}{cccccccccccccc}
        \toprule

\textbf{Param.}&\textbf{Values}&\textbf{Param.}&\textbf{Values}&\textbf{Param.}&\textbf{Values}\\
\midrule
$\boldsymbol{V_{\textbf{in}}}$
& 12\,V & $\boldsymbol{L}$ & 240\,nH & $\boldsymbol{T_{\textbf{on}}}$ 
& 100\,ns \\
\midrule
$\boldsymbol{V_{\textbf{out}}}$
& 2.2\,V & $\boldsymbol{C}$ & 100 $\mu$F & $\boldsymbol{R_s}$ & 10\,m$\Omega$\\
\bottomrule
    \end{tabular}
\end{table}

\begin{table}[tb]
    \caption{Validations of the proposed stability criterion (\ref{eqn:stability_criterion})}
    \label{table:cont_on_cm_stability_criterion}
    \centering
    \begin{tabular}{cccccccccccc}
        \toprule
\textbf{Case} & $\mathbf{\emph{R}}$ & $\mathbf{\emph{I}_{\text{cmd}}}$ & $\mathbf{|\hat{\alpha}|}$ & \textbf{Corollary 1} \cite{cmpartone2022} & \textbf{(\ref{eqn:stability_criterion})}\\
\midrule
        1 & $0.4\,\Omega$ 
        & 4.5 A
        & 0.48
        & $|\hat{\alpha}| < 0.5$ & $|\hat{\alpha}| < 0.24$ \\
\midrule
        2 & $0.05\,\Omega$ 
        & 43 A
        & 0.3
        & $|\hat{\alpha}|< 0.5$ & $|\hat{\alpha}|< 0.44$ \\
\bottomrule
    \end{tabular}
\end{table}

\begin{figure}
    \centering
    \includegraphics[width = 5cm]{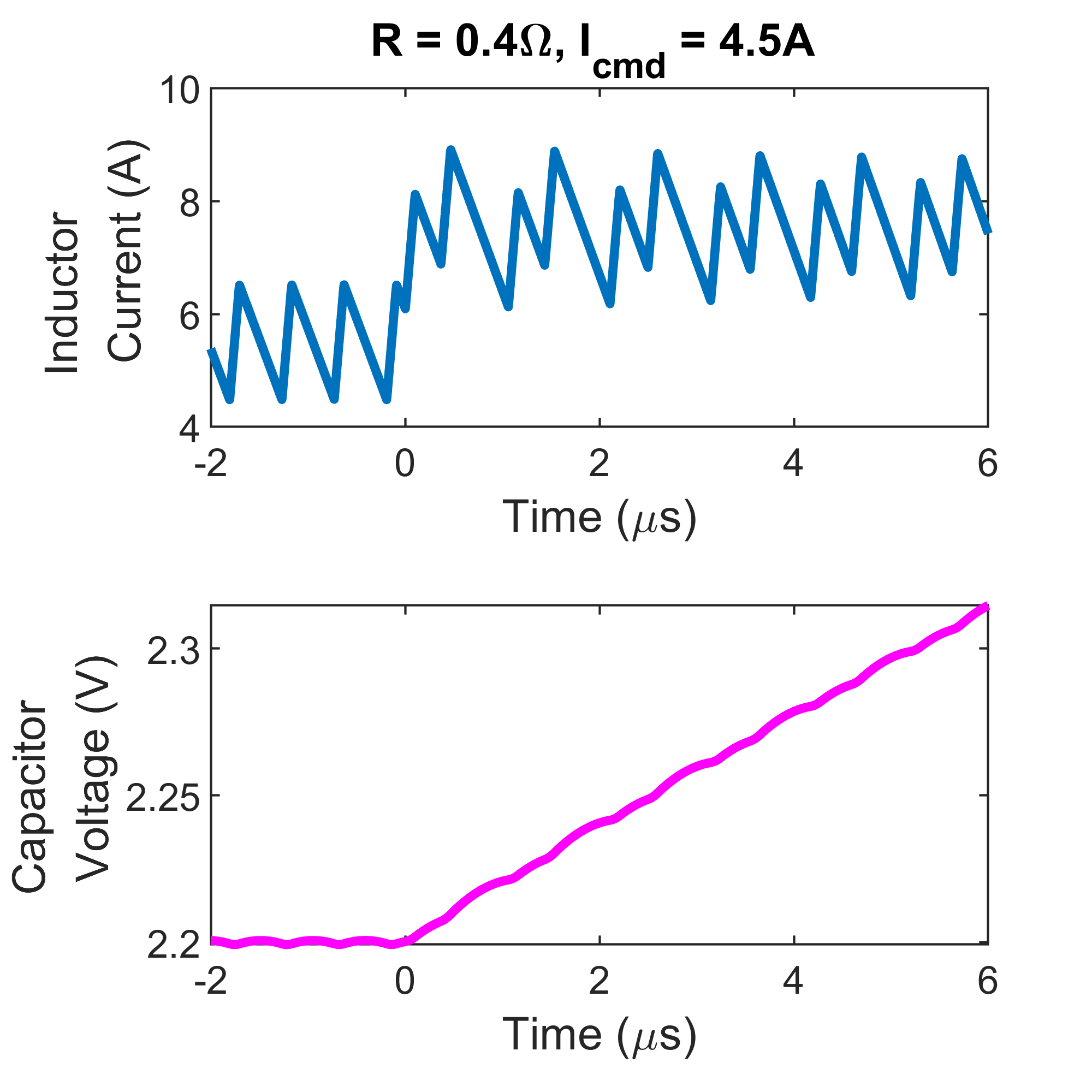}
    \caption{In case study 1, for $R = 0.4\,\Omega$ and $I_{\text{cmd}} = 4.5\,\text{A}$, according to stability criterion (\ref{eqn:stability_criterion}), the current control loop is stable if the interference is sector-bounded by $|\hat{\alpha}| < 0.24$. In the Simulink simulation, interference with sector bound $|\hat{\alpha}| = 0.48$ is added to the inductor current measurement; the inductor current waveform is {\em unstable} for a current step.
    This result shows that the expectation of stability from pre-existing theory that ignores the dependence of the current ramp on the output voltage (corollary 1 in \cite{cmpartone2022}) is inadequate.}
    \label{fig:alpha_0_48.png}
\end{figure}
\begin{figure}
    \centering
    \includegraphics[width = 5cm]{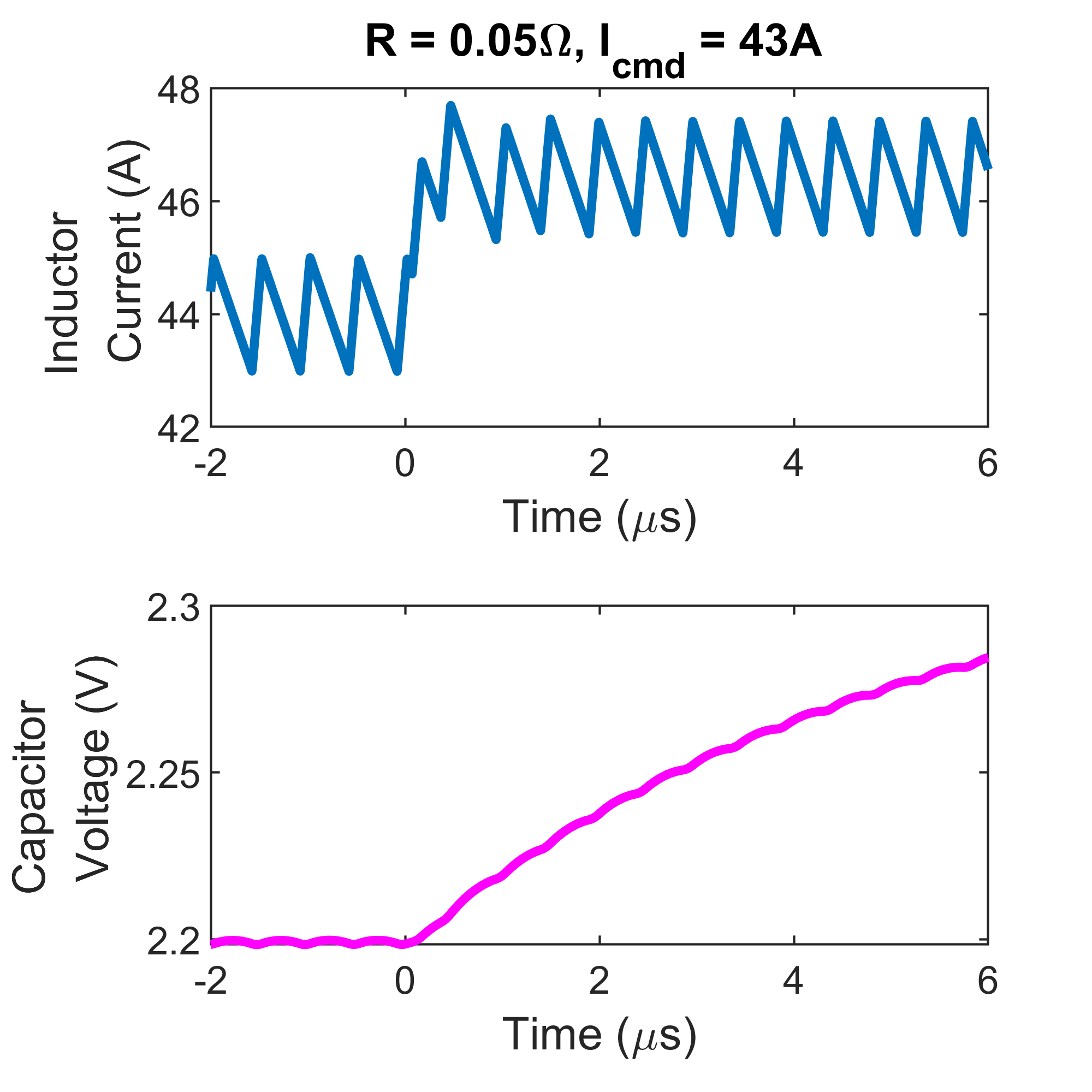}
    \caption{In case study 2, for $R = 0.05\,\Omega$ and $I_{\text{cmd}} = 43\,\text{A}$, according to the proposed stability criterion (\ref{eqn:stability_criterion}), the current control loop is stable when the interference is sector-bounded by $|\hat{\alpha}| < 0.44$. In the Simulink simulation, the interference with sector bound $|\hat{\alpha}| = 0.3$ is added to the inductor current measurement; the inductor current waveform is {\em stable} for a current step.}
    \label{fig:alpha_0_30.png}
\end{figure}

\section{Conclusion}
The theoretical contribution of this paper provides an analytical and practical stability criterion for designing current-mode dc-dc converters with large-signal stability guarantees.
The criteria indicate that the $L/R$ and $RC$ time constants are the design parameters which determine the amount of coupling between the current block and voltage block. The current block and voltage block can be decoupled by increasing $L/R$, $RC$, or $T_s^{\text{min}}$, or by decreasing $T_s^{\text{max}}$.

\section*{Acknowledgements}
Special gratitude to Professor Peter Seiler for his suggestions in the large-signal stability analysis.

\bibliographystyle{ieeetr}
\bibliography{library_fixed.bib}
\end{document}